\documentclass[10pt,letterpaper,conference,twocolumn]{IEEEtran}

\usepackage{graphicx}
\usepackage{amsmath}
\usepackage{amssymb}
\usepackage{amsxtra}
\usepackage{epsfig}
\usepackage{amsxtra}
\usepackage{amsfonts}
\usepackage{url}

\usepackage{epsf,graphicx}
\usepackage{bm}
\usepackage[active]{srcltx}

\newtheorem{thm}{Theorem}
\newtheorem{prop}{Proposition}

\newtheorem{lem}{Lemma}

\usepackage{amsmath}
\usepackage{latexsym}
\usepackage{enumerate}
\usepackage{graphics}
\usepackage{graphicx}
\usepackage{psfrag}
\usepackage{amsfonts}
\usepackage{amssymb}%
\providecommand{\U}[1]{\protect\rule{.1in}{.1in}}
\providecommand{\U}[1]{\protect\rule{.1in}{.1in}}
\providecommand{\U}[1]{\protect\rule{.1in}{.1in}}
\providecommand{\U}[1]{\protect\rule{.1in}{.1in}}

\begin{document}

\title{The Deterministic Sum Capacity of a Multiple Access Channel Interfering with a Point to Point Link}

\author{
\authorblockN{J\"org B\"uhler}
\authorblockA{Heinrich-Hertz-Lehrstuhl f\"ur Informationstheorie\\ und theoretische Informationstechnik \\
Technische Universit\"at Berlin\\
Einsteinufer 25, D-10587 Berlin, Germany\\
Email: joerg.buehler@mk.tu-berlin.de}
\and
\authorblockN{Gerhard Wunder}
\authorblockA{Fraunhofer Heinrich Hertz Institute \\
Einsteinufer 37, D-10587 Berlin, Germany\\
Email: gerhard.wunder@hhi.fraunhofer.de}
}

\maketitle
\def\thefootnote{\fnsymbol{footnote}}
\footnotetext{This research is supported by Deutsche Forschungsgemeinschaft (DFG) under grant WU 598/1-2.}

\renewcommand{\thefootnote}{\arabic{footnote}} 

\begin{abstract}
In this paper, we use  the linear deterministic approximation model to study a two user multiple access channel mutually interfering with a point to point link, which represents a basic setup of a cellular system. We derive outer bounds on the achievable sum rate and construct coding schemes achieving the outer bounds. For a large parameter range, the sum capacity is identical to the sum capacity of  the interference channel obtained by silencing the weaker user in the multiple access channel. For other interference configurations, the sum rate can be increased using interference alignment, which exploits the channel gain difference of the users in the multiple access channel. From these results, lower bounds on the generalized degrees of freedom for the Gaussian counterpart are derived. 
\end{abstract}

\section{Introduction}\label{sec:introduction} 
In recent years, approximate characterizations of capacity regions of multi-user systems have gained more and more attention, with one of the most prominent examples being the characterization of the capacity of the Gaussian interference channel to within one bit/s/Hz in \cite{EtkinTseWangIT08}. One of the tools that arised in the context of capacity approximations and has been shown to be useful in many cases is the {\em linear deterministic model} introduced in \cite{AvestimehrTseAllerton2007} \cite{AvestimehrDiggaviTse_IT2011}. Here, the channel is modeled as a deterministic mapping that operates on bit vectors and mimics the effect the physical channel and interfering signals have on the binary expansion of the transmitted symbols. Basically, the effect of the channel is to erase a certain number of ingoing bits, while superposition of signals is given by the modulo addition. 
Even though this model deemphasizes the effect of thermal noise, it is able to capture some important basic features of wireless systems, namely the superposition and broadcast properties of electromagnetic wave propagation.  Hence, in multi-user systems where interference is one of the most important limiting factors on system performance, the model can also be useful to devise effective coding and interference mitigation techniques. There are many examples where a linear deterministic analysis can be successfully carried over to coding schemes for the physical (Gaussian) models or be used for approximative capacity or (generalized) degrees of freedom determination, see for example \cite{BreslerTseETCOMM08}-\nocite{Cadambe_IT2009}\cite{HuangCadambeJafarISIT09}.


{\bf Contributions.}
From a practical viewpoint, {\em cellular systems} are of major interest. Generally, a cellular system consists of a set of base stations each communicating with a distinct set of (mobile) users. Effective coding and interference mitigation schemes are still an active area of research. Approximative models such as the linear deterministic approach might help to gain more insight into these problems. In \cite{BW_ISIT11}, the capacity of a basic cellular setup, a multiple access interfering with a point to point link, has been determined for the linear deterministic model in the case of {\em symmetric weak interference}. In this paper, we extend these results to the case of arbitrarily strong interference with respect to the achievable sum rate. Furthermore, we use these results to lower bound the generalized degrees of freedom for the corresponding Gaussian channel. 

{\bf Organization.}
The paper is organized as follows: 
Section \ref{sec:systemmodel} introduces the system model. In section \ref{sec:outerbounds}, we derive outer bounds on the 
achievable sum rate. In section \ref{sec:achievability}, we construct coding schemes that achieve these outer bounds, thereby
characterizing the sum capacity of the system. From these results, a lower bound on the the generalized degrees of freedom for the Gaussian case is derived in section \ref{sec:gendof}. Finally, section \ref{sec:conclusions} concludes the paper.

{\bf Notation.}
Throughout the paper, $\mathbb{F}_2 = \{0,1\}$ denotes the binary finite field, for which addition is written as $\oplus$, which is addition modulo 2. For two matrices $A \in \mathbb{F}_2^{n_A \times m}$ and $B \in \mathbb{F}_2^{n_B \times m}$, we denote by $[A;B] \in \mathbb{F}_2^{n_A + n_B \times m}$ the matrix that is obtained by stacking $A$ over $B$. Moreover, for a sequence of matrices $A_1, \ldots, A_n$, we write $(A_k)_{k=1}^n$ for the stacked matrix, i.e.
$(A_k)_{k=1}^n = [(A_k)_{k=1}^{n-1}; A_n]$. Similarly, $[A B]$ stands for placing $A$ next to $B$. For a matrix $A \in \mathbb{F}_2^{n \times m}$, for $1 \leq i \leq j \leq n$ the sub-matrix $A[i:j] \in  \mathbb{F}_2^{j-i+1 \times m}$ is given by taking only the rows $i$ to $j$ of the matrix $A$. For $x \in \mathbb{R}$, the positive part is denoted by $(x)^{+} = \text{max}(x,0)$. Finally, $\text{div}$ and $\text{mod}$ denote integer division and the modulo operation, respectively, where we use the convention $x \text{~div~} 0 = 0$. 

\section{System Model}\label{sec:systemmodel} 


The system we consider here represents a basic version of the uplink of cellular system and consists of three transmitters (mobile users) $Tx_{1},Tx_{2}$ and $Tx_{3}$ and two receivers (base stations) $Rx_{1}, Rx_{2}$. The system is modeled using the {\em linear deterministic model} \cite{AvestimehrTseAllerton2007, AvestimehrDiggaviTse_IT2011}. 
Here, the input symbol at transmitter $Tx_{i}$ is given by a bit vector $X_i \in \mathbb{F}_2^q$ and the output bit vectors $Y_j$ at $Rx_{j}$ are deterministic functions of the inputs: 
Defining the shift matrix $S \in \mathbb{F}_2^{q \times q}$ by
\begin{equation}
S = \begin{pmatrix} 0 & 0 & 0 & \cdots & 0\\ 1 & 0 & 0 & \cdots & 0\\ 0 & 1& 0 & \cdots &0 
\\ \vdots & \ddots& \ddots & \ddots & \vdots \\ 0 & \cdots& 0 &1 & 0 \\ 
\end{pmatrix},
\end{equation}
the input/output equations of the system are given by
\begin{IEEEeqnarray}{rCl}
Y_1 & = & S^{q-n_{11}} X_1 \oplus S^{q-n_{12}}X_2 \oplus S^{q-n_{13}}X_3, \\ \notag
Y_2 & = & S^{q-n_{21}} X_1 \oplus S^{q-n_{22}}X_2 \oplus S^{q-n_{23}}X_3. 
\end{IEEEeqnarray}
Here, $q$ is chosen arbitrarily such that $q \geq \text{max}_{i,j}\{n_{ij}\}$. Note that $n_{ij}$ gives the number of bits that can be passed over the channel between $Tx_{j}$ and $Rx_{i}$, i.e. $n_{ij}$ represent channel gains. There are three messages to be transmitted in the system: $W_{ij}$ denotes the message from transmitter $Tx_{j}$ to the intended receiver $Rx_{i}$. The definitions of (block) codes, error probability, achievable rates and the capacity region are according to the standard information-theoretic definitions. For the remainder of the paper, the transmission rate corresponding to message $W_{11}$ is represented by $R_1$, the rate corresponding to $W_{12}$ by $R_2$ and the rate for $W_{23}$ by $R_3$. The sum rate is written as $R_{\Sigma} = R_1 + R_2 + R_3$. 

In the following, we assume without loss of generality that $n_{11} \geq n_{12}$ and write $n_1 = n_{11}, n_2 = n_{12}$. Also, we let $\Delta = n_1 - n_2$. In order to keep the presentation clear and reduce the number of cases to be distinguished, we make some further assumptions on the channel gains: We let $n_{23} = n_1$ and $n_i := n_{21} = n_{22} = n_{13}$. We remark that these restrictions may seem unrealistic, but can easily be removed and the techniques applied in the following extend to more general cases as well. 

The corresponding (real) Gaussian channel is defined by output symbols
\begin{IEEEeqnarray}{rCl} \label{eq:GaussianChannel}
Y_1 &=& h_{1} X_{1} + h_{2}X_2 + h_{i}X_3 + Z_1,\\ \notag
Y_2 &=& h_{i} X_{1} + h_{i}X_2 + h_{1} X_3 + Z_2
\end{IEEEeqnarray}
with (non-varying) channel coefficients $h_1,h_2,h_i \in \mathbb{R}_{\geq 0}$, input symbols $X_i$ subject to power constraints $\mathbf{E}[X_i^2] \leq P_i$ and additive white Gaussian noise $Z_i \sim \mathcal{N}(0,1)$. 

We remark that a related model with identical channel gains in the two-user cell ($h_1 = h_2$) is studied in \cite{ChaabanSezgin2010}. In this work, a capacity outer bound is derived, which is achievable in some cases. 
For the {\em weak interference} case, for which $2 n_i \leq n_2$, the capacity region for the deterministic model has been characterized in \cite{BW_ISIT11}. Here, the maximum sum rate is given by 
\begin{equation}
R_{\Sigma} \leq  n_1 + n_2 - 2n_i + \varphi_1(n_i,\Delta),
\end{equation}
where the function $\varphi_1$ is defined in (\ref{eq:phi1}) below. In the following section, we derive outer bounds on the achievable sum rate for scenarios with arbitrarily strong interference.

\section{Outer bounds on sum rate}\label{sec:outerbounds} 
Before stating the outer bounds, we introduce some definitions. Throughout the following, we let $\alpha = \frac{n_i}{n_1}, \beta = \frac{n_2}{n_1} \leq 1$, $\sigma = 2n_i - n_1$ and $\tau = 2n_i - n_2$. Furthermore, for $p,q \in \mathbb{R}_{\geq 0}$, we define $l(p,q) = \left \lfloor\frac{p}{q}\right\rfloor$ for $q > 0$ and $l(p,0) = 0$. Moreover, we define the 
functions $\varphi_1,\varphi_2: \mathbb{R}_{\geq 0} \times \mathbb{R}_{\geq 0} \rightarrow \mathbb{R}$ as (note that $0$ is considered as an even number)
\begin{equation} \label{eq:phi1}
\varphi_1(p,q) = \begin{cases} q + \frac{l(p,q) q}{2}, & \mbox{if} \;\; l(p,q) \;\; \mbox{is even} \\ p - \frac{(l(p,q)-1)q}{2}, & \mbox{if} \;\; l(p,q) \;\; \mbox{is odd}, \end{cases}
\end{equation}
\begin{equation} \label{eq:phi2}
\varphi_2(p,q) = \begin{cases} p - \frac{l(p,q) q}{2}, & \mbox{if} \;\; l(p,q) \;\; \mbox{is even} \\ \frac{(l(p,q)+1)q}{2}, & \mbox{if} \;\; l(p,q) \;\; \mbox{is odd}. \end{cases}
\end{equation}
Moreover, let $\overline{\alpha} = \text{min}\left(1-\frac{\beta}{2},\frac{2}{3}\right)$. Then, we have the following outer bound on the achievable sum rate:
\begin{prop} \label{prop:outerbounds}
For $\alpha < \beta$, it holds that 
 \begin{equation}
R_{\Sigma} \leq \begin{cases} n_1 + n_2 - 2n_i  + \varphi_1(n_i,\Delta) &\alpha \in \left[0, \frac{1}{2}\right] \\ 
2n_i + \varphi_2(n_2-n_i,\Delta), &  \alpha \in \left(\frac{1}{2},\overline{\alpha}\right) \\
2n_i + \varphi_2(2n_1- 3n_i,\Delta), &  \alpha \in \left[\overline{\alpha},\frac{2}{3}\right) \\
\text{min}\left(2n_1,\text{max}(n_1,n_i) + (n_1-n_i)^{+}\right), & \alpha \in \left[\frac{2}{3},\infty\right)\\
\end{cases} 
\end{equation}
For $\alpha \geq \beta$, 
\begin{IEEEeqnarray}{l}
R_{\Sigma} \leq \text{min}\left(\text{max}(n_1,n_i) + (n_1-n_i)^{+}, \right.\\ \notag
\quad \quad  \quad \quad \,\,\,\,\,  \left. 2 \cdot \text{max}(n_i,(n_1-n_i)^{+})\right).
\end{IEEEeqnarray}
 \end{prop}
We remark that the bounds for $\alpha \geq \beta$ and $\alpha < \beta,\alpha \geq \frac{2}{3}$ coincide with the sum capacity for the interference channel formed by $Tx_1,Tx_3$ and $Rx_1,Rx_2$ (also see section \ref{sec:achievability}). 
\begin{proof}

In order to prove the outer bounds on the sum rate, we first apply Fano's inequality: For each triple of achievable rates $(R_1,R_2,R_3)$, Fano's inequality implies that there exists a sequence
$\varepsilon_{N}$ with $\varepsilon_{N} \rightarrow 0$ for $N \rightarrow \infty$ and
a sequence of joint factorized distributions $p(x_1^N)p(x_2^N)p(x_3^N)$ such that for all $N \in \mathbb{N}$
$R_1 + R_2 \leq \frac{1}{N}I(X_1^N,X_2^N;Y_1^N) + \varepsilon_{N}$ and $R_3 \leq \frac{1}{N}I(X_3^N;Y_2^N) + \varepsilon_{N}$.
\begin{figure}
\begin{center}
\includegraphics[scale=0.46]{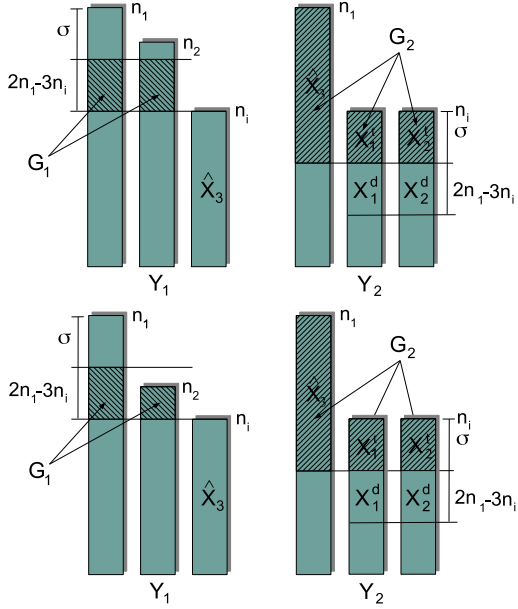}
\caption{Illustration of received signals and the signal parts used for upper bounding the sum rate. Here, the bars represent the bit vectors as seen at the two receivers (for $N=1$);  the zero parts due to the channel shifts are not displayed. The upper figure represents the case $\alpha \in \left[\overline{\alpha},\frac{2}{3}\right)$ and the lower figure the case $\alpha \in \left(\frac{1}{2},\overline{\alpha}\right)$ (for $\alpha < \beta$). }
\label{fig:geniediagram}
\end{center} 
\end{figure}
The upper bound obtained in this way is bounded further in different ways for the different channel gain regimes
determined by $\alpha$ and $\beta$. 

{\bf i) The case $\alpha < \beta$:} We start with the case $\alpha \geq \frac{2}{3}$. Here, the bound is obtained by letting a genie provide 
$X_1^N$ and $X_2^N$ to receiver $Rx_2$, which results in
\begin{IEEEeqnarray}{l} \label{eq:alphagreater23bound}
(R_1 + R_2 + R_3)N -\varepsilon_{N}N \\ \notag
\,\, \,\,\, \leq  I(X_1^N,X_2^N;Y_1^N) + I(X_3^N;Y_2^N,X_1^N,X_2^N) \\ \notag
\,\, \,\,\, = H(Y_1^N) + H(X_3^N) - H(Y_1^N|X_1^N,X_2^N)\\ \notag
\,\, \,\,\, \leq N \text{min}\left(2n_1,\text{max}(n_1,n_i) + (n_1-n_i)^{+}\right).
\end{IEEEeqnarray}

For $\alpha \in \left(\frac{1}{2},\frac{2}{3}\right)$, we write  $G_1 = Y^N_1[\sigma+1:\sigma+\tau]$, $G_2 = Y^N_2[1:n_i]$, $\hat{X_3} = X_3^N[1:n_i]$, $X_1^d = X_1^N[\sigma+1:n_1-n_i],X_2^d = X_2^N[\sigma+1:n_1-n_i]$, $X_1^t = X_1^N[1:\sigma]$ and $X_2^t = X_2^N[1:\sigma]$ (c.f. Figure \ref{fig:geniediagram}). Then, we have 
\begin{IEEEeqnarray}{l}
(R_1 + R_2 + R_3)N -\varepsilon_{N}N \\ \notag
\,\, \,\,\, \leq  I(X_1^N,X_2^N;Y_1^N,G_1) + I(X_3^N;Y_2^N,G_2) \\ \notag
\,\, \,\,\, = H(G_1) + H(Y_1^N|G_1) - H(Y_1^N|G_1,X^N_1,X^N_2) + H(G_2)\\ \notag
\,\,\,\,\,\,\,\,\,\,\, - H(G_2|X^N_3)+ H(Y^N_2|G_2) - H(Y^N_2|G_2,X^N_3)\\ \notag
\,\, \,\,\, = H(G_1) + H(Y^N_1|G_1) - H(\hat{X_3}) + H(G_2) \\ \notag
\,\,\,\,\,\,\,\,\,\,\,- H(X_1^t \oplus X_2^t) + H(Y^N_2|G_2) - H(Y^N_2|G_2,X^N_3)\\ \notag
\,\, \,\,\,  \leq 2 n_i N+ H(G_1) - H(X_1^d \oplus X_2^d),
\end{IEEEeqnarray}
where we used $H(Y_1|G_1) \leq N(\sigma + n_i)$, $H(Y_2|G_2) \leq N(n_1-n_i)$,
$H(G_2) - H(\hat{X_3}) \leq H(\hat{X_3}) + H(X_1^t \oplus X_2^t) - H(\hat{X_3}) = H(X_1^t \oplus X_2^t)$ and
$H(Y^N_2|G_2,X^N_3) \geq H(X_1^d \oplus X_2^d)$.
%
Now 
\begin{equation}
 H(G_1) - H(X_1^d \oplus X_2^d) \leq
\begin{cases} 
 N \varphi_2(2n_1- 3n_i,\Delta) , & \mbox{if} \;\;   \sigma \geq \Delta \\
 N \varphi_2(n_2 - n_i,\Delta), & \mbox{if} \;\;  \sigma < \Delta,  
\end{cases}
\end{equation}
which for $\sigma \geq \Delta$ follows from applying Lemma \ref{lem:entropybound1} given below. A slight modification of Lemma \ref{lem:entropybound1} (which is not given here due to the lack of space) shows the other case. Note that $\sigma < \Delta$ corresponds to $\alpha \in \left(\frac{1}{2},\overline{\alpha}\right)$ and $\sigma \geq \Delta$ to $\alpha \in \left[\overline{\alpha},\frac{2}{3}\right)$.

The bound for $\alpha \in \left[0,\frac{1}{2}\right)$ is proved by a similar argument, which we sketch only here: $G_1$ and $G_2$ are taken as $G_1 = Y^N_1[1:n_1-n_i]$ and $G_2 = Y^N_2[1:n_i]$, respectively and by an appropriate adjustment of Lemma \ref{lem:entropybound1}, the bound follows. We remark that for $\alpha \in \left[0, \frac{\beta}{2}\right)$, the outer bound also follows from the results in \cite{BW_ISIT11}. 

{\bf ii) The case $\alpha \geq \beta$:} The bound $R_{\Sigma} \leq 2 \cdot \text{max}(n_i,(n_1-n_i)^{+})$ follows easily by providing the genie information $G_1 = X_1^N[1:\text{min}(n_1,n_i)]$ and $G_2 = X_1^N[1:\text{min}(n_1,n_i)]$ to receiver $Rx_1$ and $Rx_2$, respectively. The other bound $R_{\Sigma}\leq \text{min}\left(2n_1,\text{max}(n_1,n_i) + (n_1-n_i)^{+}\right)$ has already been shown above in (\ref{eq:alphagreater23bound}). 
\end{proof}

\begin{lem}\label{lem:entropybound1}
Let $A \in \mathbb{F}_2^{n\times m}$, $B \in \mathbb{F}_2^{n + \Delta \times m}$ be independent random matrices with $m,n, \Delta \in \mathbb{N}$ and  $B' = B [1:n], B'' = B[\Delta+1:n+\Delta]$. 
Then, it holds that
\begin{equation} \label{eq:shiftbound1}
H(A \oplus B' ) - H(A \oplus B'') \leq m \varphi_2(n,\Delta).
\end{equation}
\end{lem}

\begin{proof}
We let $l = n ~\text{div}~ \Delta, Q = n~ \text{mod}~ \Delta$ and introduce the following labels for blocks of rows of the matrices $A, B'$ and $B''$: $A = [(A_k)_{k=1}^{l};Q_{A}], B' = [(B_k)_{k=0}^{l-1};Q'_{B}], B'' = [(B_k)_{k=1}^{l};Q''_{B}]$, where $Q_{A},Q'_{B},Q''_{B} \in \mathbb{F}_2^{Q \times m}$ and $A_k,B_k \in \mathbb{F}_2^{\Delta \times m}$. 
First consider the case that $l$ is even. Then we have
\begin{IEEEeqnarray}{l}
H(A \oplus B' ) - H(A \oplus B'') \\ \notag
\,\, \,\,\,\leq m\Delta + H\left[(A_{k+1} \oplus B_k)_{k=
1}^{l-1};Q_A \oplus Q'_B \right] \\ \notag
\,\,\,\,\,\,\,\,\,\,\,- H\left[(A_k \oplus B_k)_{k=1}^{l}; Q_A \oplus Q''_{B}\right] \\ \notag
\,\, \,\,\, \leq m\Delta + H\left[(A_{k+1} \oplus B_k)_{k=1}^{l-1};Q_A \oplus Q'_B \right]\\ \notag
\,\,\,\,\,\,\,\,\,\,\,- H\left[(A_k \oplus B_k)_{k=1}^{l} \left|(A_{2k-1})_{k=1}^{l/2},(B_{2k})_{k=1}^{l/2}\right.\right]
\\ \notag
\,\, \,\,\,= m\Delta + H\left[(A_{k+1} \oplus B_k)_{k=1}^{l-1};Q_A \oplus Q'_B \right] \\ \notag
\,\,\,\,\,\,\,\,\,\,\,- H\left[(A_{2k})_{k=1}^{l/2}; (B_{2k-1})_{k=1}^{l/2}\right]
\\ \notag
\,\, \,\,\,\leq m \frac{l\Delta}{2} + mQ + H\left[\left(A_{2k} \oplus B_{2k-1}\right)_{k=1}^{l/2}\right]\\ \notag
\,\,\,\,\,\,\,\,\,\,\, - H\left[(A_{2k})_{k=1}^{l/2}; (B_{2k-1})_{k=1}^{l/2}\right] \\ \notag
\,\, \,\,\,\leq m  \frac{l\Delta}{2} + mQ = m \varphi_2(n,\Delta),
\end{IEEEeqnarray}
where the last inequality is due to the independence of $A$ and $B$.

A similar line of argument can be applied for the case that $l$ is odd. 
\end{proof}

\section{Achievability}\label{sec:achievability} 
In this section, we describe how the outer bounds  given in Proposition \ref{prop:outerbounds} can be achieved. It turns out that it suffices to restrict to {\em linear coding} over a single symbol period. For this, each transmitter $Tx_{i}$ chooses a precoding matrix $V_i \in \mathbb{F}_2^{k_i \times q}$
and transmits the signal $V_i x_i$, where $x_i \in \mathbb{F}_2^{k_i \times 1}$ contains the data. In general, 
writing $A = V_1$, $B = S^{q-n_2} V_2$,
$C = S^{q-n_i}V_3$, $D = S^{q-n_i}V_1$, $E = S^{q-n_i}V_2$ and $F = V_3$, it is easy to see that the following rates are achievable under linear precoding: \begin{IEEEeqnarray}{rCl}
R_1 & =&  \text{rank}([A ~ B ~ C]) - \text{rank}([B ~ C]), \\
R_2 & =&  \text{rank}([A ~ B ~ C]) - \text{rank}([A ~  C]), \\
R_3  &= &  \text{rank}([D ~  E~  F]) - \text{rank}([D ~ E]).
\end{IEEEeqnarray}
In the following, we sketch how to construct precoding matrices $V_i$ that achieve the outer bounds. 
The construction of the precoding matrices again is different for different channel parameter configurations.

For the cases $\alpha \geq \beta$ and $\alpha < \beta, \alpha \geq \frac{2}{3}$, the maximum sum-rate can be achieved by keeping transmitter $Tx_2$ silent and applying the sum-rate optimal interference channel code for the interference channel formed by $Tx_1,Tx_3$ and $Rx_1,Rx_3$. The results from \cite{BreslerTseETCOMM08} show that in this way, the bounds can be achieved (by linear coding). 

For $\alpha < \beta, \alpha \in [0, \frac{\beta}{2}]$, an optimal construction follows from the results in \cite{BW_ISIT11}, and the extension to the case $\alpha < \beta, \alpha \in (\frac{\beta}{2}, \frac{1}{2}]$ is straightforward. Here, the maximum sum rate can be achieved by orthogonal coding, where for each transmitter, a set of bit levels to be used for data transmission is specified such that at the intended receiver, there is no overlap of these levels with levels used by any other transmitter. The assignment can be interpreted as {\em interference alignment}, where the bit levels are chosen such that the interference caused by $Tx_1$ and $Tx_2$ aligns at $Rx_2$ as much as possible in the levels that are unused by $Tx_3$. 

For the remaining case $\alpha \in \left(\frac{1}{2},\frac{2}{3}\right)$, the maximum sum rate can not be obtained by orthogonal coding; instead, {\em coding across levels} is necessary (as is for the interference channel in a certain interference range \cite{BreslerTseETCOMM08}). However, interference alignment still plays a key role. To describe the construction, 
we let $\rho = n_i - \sigma - \tau$ in the following. For the construction, three cases have to be distinguished: a) $\rho < 0$, b) $\rho \geq 0$ and $\alpha \in \left(\frac{1}{2},\overline{\alpha}\right)$ and c) $\rho \geq 0$ and $\alpha \in \left[\overline{\alpha},\frac{2}{3}\right)$. Due to space limitations, we only consider case c) here; the constructions for a) and b) are similar. For $n, \Delta \in \mathbb{N}$, we let $ \mathcal{K}^{\Delta}_{n} = \left\{k \in \{1,\ldots,n\}: \text{mod}(k,2\Delta) < \Delta\right\}$
and for an element $k \in \mathcal{K}^{\Delta}_{n}$, $i^{\Delta}_{n}(k)$ denotes the position of the element in the sequence of increasingly ordered elements in $\mathcal{K}^{\Delta}_{n}$. Then,  precoding matrices achieving the outer bound can be constructed as follows:
\begin{IEEEeqnarray}{rCl} \label{eq:optimalcodinrange4}
V_1(k,k)  &=&   1 , 1 \leq k \leq \Delta  \\ \notag
V_1(n_1-n_i+k,k) &=&   1 , 1 \leq k \leq \Delta  \\ \notag
V_1(\tau+k,\Delta + i^{\Delta}_{\rho}(k)) & =&   1, k \in \mathcal{K}^{\Delta}_{\rho} \\ \notag
V_1(n_i + \Delta + k,\Delta+|\mathcal{K}^{\Delta}_{\rho}|+k)  &=&  1 , 1 \leq k \leq n_2-n_i \\ \notag
V_2(k,k) & = &  1, 1 \leq k \leq \sigma  \\ \notag
V_2(n_1-n_i+k,k) & = &  1, 1 \leq k \leq \sigma  \\ \notag
V_2(\tau + k,\sigma + i^{\Delta}_{\rho - \Delta}(k)) & = &  1, k \in  \mathcal{K}^{\Delta}_{\rho - \Delta} \\ \notag
V_3(k,k) & =&   1 , 1 \leq k \leq \tau  \\ \notag
V_3(n_i +k,\tau + i^{\Delta}_{\rho + \Delta}(k)) & =&   1,  k \in \mathcal{K}^{\Delta}_{\rho + \Delta}   \\ \notag
V_3(2(n_1-n_i)+k,\tau +  |\mathcal{K}^{\Delta}_{\rho + \Delta}|+ k) & =&   1 , 1 \leq k \leq \sigma   
\end{IEEEeqnarray}

An example for this coding scheme is given in Figure \ref{fig:achievabilityexample} for $n_1 = 23, n_2 = 21$ and $n_i = 13$. 

We summarize this section in the following
\begin{thm} \label{thm:sumcapacity}
The sum capacity for the linear deterministic model equals the expressions given in Proposition \ref{prop:outerbounds}.  
 \end{thm}

\begin{figure}
\begin{center}
\includegraphics[scale=0.54]{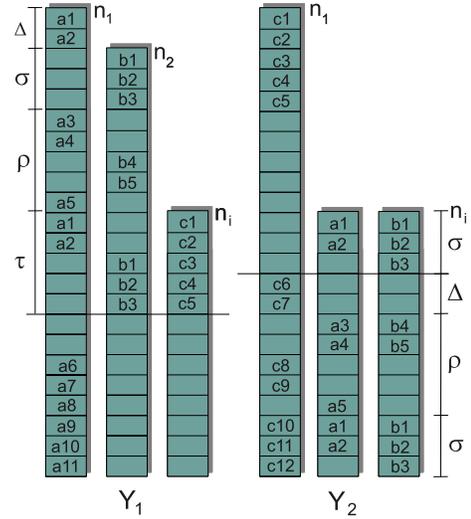}
\caption{Illustration of the code resulting from the precoding matrices given in (\ref{eq:optimalcodinrange4}) for $n_1 = 23, n_2 = 21, n_i = 13$, achieving the outer bound $R_{\Sigma} = 28$ (with $R_1 = 11, R_2=5, R_3 =12$). Here, the bars represent the bit vectors as seen at the two receivers and $a_i,b_i$ and $c_i$ are the data bits for $Tx_1,Tx_2$ and $Tx_3$, respectively. }
\label{fig:achievabilityexample}
\end{center} 
\end{figure}



\section{Generalized degrees of freedom}\label{sec:gendof} 
In this section, we briefly turn to the Gaussian channel given in  (\ref{eq:GaussianChannel}). It can equivalently be written as
\begin{IEEEeqnarray}{rCl}
Y_1 &=& \sqrt{p}X_{1} +\sqrt{p^{b}}X_2 + \sqrt{p^{a}}X_3 + Z_1,\\ \notag
Y_2 &=& \sqrt{p^{a}} X_{1} + \sqrt{p^{a}}X_2 +  \sqrt{p} X_3 + Z_2
\end{IEEEeqnarray}
with $p, a,b \in \mathbb{R}_{\geq 0}$ and with power constraints $P_i = 1$. Then, the {\em generalized degrees of freedom}  
\cite{EtkinTseWangIT08} are defined as 
\begin{equation}
 d(a,b) = \limsup_{p \rightarrow \infty} \frac{C_{\Sigma}(p,a,b)}{\frac{1}{2}\log(p)},
\end{equation}
where $C_{\Sigma}(p,a,b)$ is the sum capacity of the channel. This measure represents a high SNR description
of the system, where the channel gains are kept in constant relation (determined by $a$ and $b$) in the dB scale, which
allows a more detailed asymptotic description of the system as opposed to the degrees of freedom characterization  (also see e.g. \cite{BreslerTseETCOMM08}). 

\begin{figure}
\begin{center}
\includegraphics[scale=0.45]{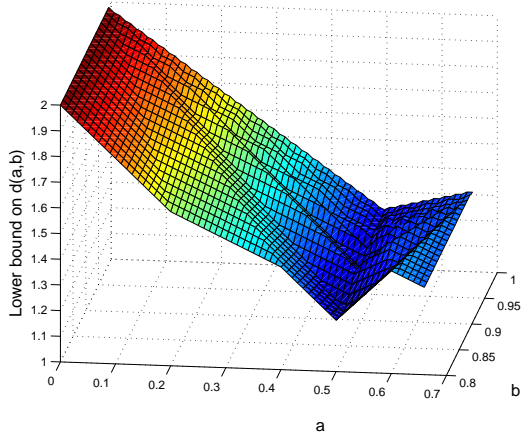}
\caption{Achievable generalized degrees of freedom for a range of parameters $a, b$.}
\label{fig:genDOF}
\end{center} 
\end{figure}
\begin{figure}
\begin{center}
\includegraphics[scale=0.45]{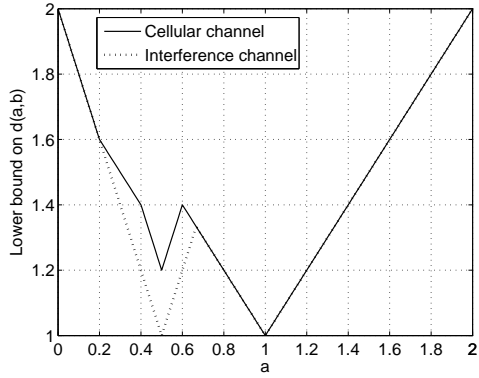}
\caption{Achievable generalized degrees of freedom for $b = 0.8$.}
\label{fig:WCurve}
\end{center} 
\end{figure}

The technique given in \cite{Jafar_Vishwanath_IT2010} \cite{Cadambe_IT2009} (and also applied in e.g. \cite{HuangCadambeJafarISIT09}) allows to transfer an achievable scheme for the deterministic model to a coding scheme for the Gaussian channel. The basic idea is to inscribe the code obtained for the linear deterministic model into a $Q$-ary expansion of the transmit signals, asymptotically mimicking the behavior of the deterministic channel. This technique can also be applied for the channel at hand. Omitting the details, we just remark that  here we use the fact that for channel gain approximation sequences $\frac{a_k}{t_k}  \rightarrow a$ and $\frac{b_k}{t_k}   \rightarrow b$ for $k \rightarrow \infty$ (where $a_k$,$b_k$, $t_k \in \mathbb{N})$, we have that $\frac{1}{t_k}\varphi_1(a_k,t_k-b_k) \rightarrow \varphi_1(a,1-b)$ for $k \rightarrow \infty$. Similarly,
 $\frac{1}{t_k}\varphi_2(b_k-a_k,t_k-b_k) \rightarrow \varphi_2(b - a,1-b)$ and
$\frac{1}{t_k}\varphi_2(2t_k-3a_k,t_k-b_k) \rightarrow \varphi_2(2 - 3a,1-b)$. Using the sum capacity results from Theorem \ref{thm:sumcapacity}, one obtains the following lower bound on the generalized degrees of freedom:

\begin{prop} \label{prop:genDOF}
For $a < b$, it holds that 
 \begin{equation}
d(a,b)\geq \begin{cases}  1+ b - 2a + \varphi_1(a,1-b), &a \in \left[0, \frac{1}{2}\right] \\ 
2a + \varphi_2(b - a,1-b), &  a \in \left(\frac{1}{2},\overline{a}\right) \\
2a + \varphi_2(2 - 3a,1-b), &  a \in \left[\overline{a},\frac{2}{3}\right) \\
\text{min}\left(2,\text{max}(1,a) + (1-a)^{+}\right),  & a \in \left[\frac{2}{3},\infty\right).\\
\end{cases} 
\end{equation}
For $a \geq b$,
 \begin{equation}
d(a,b) \geq \text{min}\left(\text{max}(1,a) + (1-a)^{+}, 2 \,\text{max}(a,(1-a)^{+})\right). 
\end{equation}
 \end{prop}

Figure \ref{fig:genDOF} displays the lower bound on $d(a,b)$ in the range $a \in [0,0.7], b \in [0.8,1]$. For $b = 0.8$, the achievable generalized degrees of freedom are shown in Figure \ref{fig:WCurve} for different $a$ values, together with
the generalized degrees of freedom for the interference channel consisting of only $Tx_1,Tx_3$ and $Rx_1,Rx_2$. Note that the latter one represents the well-known {\em W curve}  \cite{EtkinTseWangIT08} of the generalized degrees of freedom for the interference channel. For $a < \frac{2}{3}$, the channel gain difference in the two-user cell can be exploited for interference alignment, pushing the achievable generalized degrees of freedom higher than the W curve, whereas for  $a \geq \frac{2}{3}$, the lower bound can be achieved by coding only for the interference channel consisting of $Tx_1,Tx_3$ and $Rx_1,Rx_2$. 

\section{Conclusions}\label{sec:conclusions} 
In this paper, we studied the linear deterministic model for a cellular-type channel where a two user multiple access channel mutually interferes with a point to point link. Under certain symmetry assumptions on the channel gains, we derived the sum capacity and the corresponding transmission schemes, which use interference alignment and linear coding across bit levels. While for a large parameter range, the 
sum capacity is identical to the sum capacity of  the interference channel obtained by silencing the weaker user in the two-user cell (multiple access channel), for a certain parameter range, the channel gain difference of in the two-user cell allows to get a higher sum rate using interference alignment. Finally, from these results, a lower bound on the generalized degrees of freedom for the Gaussian channel was given, increasing the W curve for the interference channel in a certain interference range.

Although we have considered a restricted setup in this paper, we believe that the achievability and converse arguments used in this paper give valuable insights for the consideration of more general systems. Future work will study extensions to the more general cases with additional users. Another interesting direction for future investigations is to further explore the connections to the Gaussian equivalent of the channel, specifically concerning outer bounds on the generalized degrees of freedom of the system and approximate capacity characterizations.

\bibliographystyle{IEEEtran}

\end{document}